\documentclass[11pt]{article}
\usepackage{amsfonts,amssymb,amsmath,amsthm,bbm,mathrsfs,mathtools}
\usepackage{graphicx, color, psfrag,hyperref,xcolor, plain}

\usepackage{authblk}

\usepackage[english]{babel}
\usepackage{color}

\newtheorem{theorem}{Theorem}

\newtheorem{proposition}{Proposition}
\newtheorem{lemma}{Lemma}

\def\exp{\mathop{\textrm{\rm exp}}\nolimits}              

\newcommand\blue[1]{\textcolor{blue}{}}

\newcommand{\be}{\begin{equation}}
\newcommand{\ee}{\end{equation}}
\newcommand{\vertiii}[1]{{\left\vert\kern-0.25ex\left\vert\kern-0.25ex\left\vert #1 
    \right\vert\kern-0.25ex\right\vert\kern-0.25ex\right\vert}}

\DeclarePairedDelimiter\abs{\lvert}{\rvert}%

\makeatletter
\let\oldabs\abs
\def\abs{\@ifstar{\oldabs}{\oldabs*}}

\newcommand{\E}{\mathbb E}

\newcommand{\scrf}{\ensuremath{\mathscr{F}}}
\newcommand{\scre}{\ensuremath{\mathscr{E}}}

\newcommand{\si}{\ensuremath{\sigma}}

\newcommand{\lf}{\ensuremath{\left(}}
\newcommand{\ri}{\ensuremath{\right)}}

\newcommand*\diff{\mathop{}\!\mathrm{d}}

\def\Var{\mathop{\textrm{\rm Var}}\nolimits}                  

\begin{document}

\title{{\bf Local Central Limit Theorem for Long-Range Two-Body Potentials at Sufficiently High Temperatures}}
\author[1]{Eric O. Endo  \footnote{ericossamiendo@gmail.com}}
\author[1,2]{Vlad Margarint \footnote{vldu7137@colorado.edu}}
\affil[1]{NYU-ECNU Institute of Mathematical Sciences, NYU Shanghai, 3663 Zhongshan Road North, Shanghai, 200062, China}
\affil[2]{University of Colorado Boulder, 2300 Colorado Avenue, Boulder, 80309-0395, USA}
\date{}                     
\setcounter{Maxaffil}{0}
\renewcommand\Affilfont{\itshape\small}

\maketitle

\begin{center}
{\bf Abstract}  
\end{center}
{\footnotesize
 Dobrushin and Tirozzi \cite{DT} showed that, for a Gibbs measure with the finite-range potential, the Local Central Limit Theorem is implied by the Integral Central Limit Theorem.
 Campanino, Capocaccia, and Tirozzi \cite{CCT} extended this result for a family of Gibbs measures for long-range pair potentials satisfying certain conditions. We are able to show for a family of Gibbs measures for long-range pair potentials not satisfying the conditions given in \cite{CCT}, that at sufficiently high temperatures, if the Integral Central Limit Theorem holds for a given sequence of Gibbs measures, then the Local Central Limit Theorem also holds for the same sequence. We also extend \cite{CCT} when the state space is general, provided that it is equipped with a finite measure.}

\vspace{4cm}



\section{Introduction}

The Central Limit Theorem is one of the fundamental results of Probability Theory. It states that under certain conditions, the properly scaled sum of independent and identically distributed (i.i.d) random variables converges, as the number of terms tends to infinity to a Normal random variable. However, one aspect that is missing from the Central Limit Theorem is, for instance, the rate of convergence to the limiting Normal distribution. One way to generalize the Central Limit Theorem is the celebrated Berry-Esseen Theorem \cite{Berry, Ess} that provides a more quantitative statement, namely providing the rate at which the convergence takes place.

The Central Limit Theorem was studied for random fields generated by models coming from Statistical Mechanics (see \cite{ANP,Bolt, CoxGold, Coxg, DeConnick, FKT, Geo, GH, Hegerfeldt, Iago, Iago2,  Malyshev, Nah, Nak, Neader2, Neaderhouser1, Neaderhouser2, Neader3, NewmanWright, NewmanWright2,RR,Wu} for various techniques of proof). The fact that the Central Limit Theorem holds for models satisfying the FKG inequality and that have a finite susceptibility was proved in \cite{New}, extending the work in \cite{Newold} that considers only monotonic functionals of the random variables. K\"unsch studied the model in more generality in \cite{Kun}, and he provided applications of the Central Limit Theorem and the second derivative of the pressure. That work provides an example on which it is known that the Central Limit Theorem holds in more generality than the case of discrete-valued spins considered in \cite{CCT}. Moreover, Central Limit Theorems for the Ising ferromagnet in two or more dimensions are obtained in \cite{GalandJona, GalandM, MLof, Pick}. 

Local Central Limit Theorems are also a way to provide a more refined result than the Central Limit Theorem.  Their importance in Statistical Mechanics was noticed in the study of a family of models (see \cite{CCT, CampDeg, Del Grosso, DT}) in which the authors prove that assuming the Central Limit Theorem holds for a random field defined on $\mathbb{Z}^d$, then the Local Central Limit Theorem will hold as well. 
In particular, in \cite{DT} the authors proved it for the short-range potentials, while \cite{CCT} showed it for a family of long-range potentials. Local Central Limit Theorems are important in Statistical Mechanics since, with their help, one can deduce the equivalence of the Canonical Ensemble and the Grand-Canonical Ensemble for
spin systems and for particle systems (see \cite{DT,CampDeg} for more details).

In this paper, 
we consider a sufficiently high temperature regime, and we prove that for Gibbs fields with spins taking values in a measurable set $E$, equipped with a measure $\lambda$, with the condition $\lambda(E) < \infty$, for which the Central Limit Theorem is satisfied, then the Local Central Limit Theorem also holds. Our result complements \cite{CCT} where some families of absolutely summable long-range potentials that fail the condition in \cite{CCT} still satisfy the result at sufficiently high temperatures. We also extend \cite{CCT} where only discrete value spins are considered.

Our proof is similar to the proof in \cite{CCT}. It uses the study of the characteristic functions of the random field for small and big values of the parameter. The analysis of these cases is done in two separate lemmas in Section \ref{proofofthe mainresult} that use different techniques. They both rely on the analysis of the cluster expansion done in Section \ref{ce}. The main difference is the technique to show the absolute convergence of the cluster expansion, since this method is model dependent. While in \cite{CCT} they construct walks from polymers (see Theorem 3.2 in \cite{CCT}), we adapted the proof from \cite{FV}, Chapter 5, for sufficiently high temperatures. See \cite{FP} for results about convergence of cluster expansions.

As an application of our main result, we consider the one-dimensional long-range Ising models with polynomially decaying interaction $J_{xy}=\abs{x-y}^{-2+\alpha}$, with $0\le \alpha<1$. Because of this particular long-range interaction, these models undergo a phase transition at low temperatures, and the conditions in \cite{CCT} fail for them. As a second application of our main result, we have obtained that for the models considered in \cite{Kun} for which the Central Limit Theorem holds, we have that the Local Central Limit Theorem holds as well at sufficiently high temperatures.

The paper is divided into several sections, the first one being the Introduction. Section \ref{themodels} introduces the models, the Local Central Limit Theorem, our main result, and some applications. In Section \ref{ce} we perform the cluster expansion corresponding to our model and obtain the absolute convergence of the corresponding series at sufficiently high temperatures. In Section \ref{proofofthe mainresult} we prove our main result. Section \ref{proofofthe mainresult} is divided into two further subsections, each of them proves Theorem \ref{main} assuming the condition (\ref{hightemp}) and (\ref{samecampanino}) given in the theorem.

\section{The Models}\label{themodels}

\subsection{Definitions and Notation}
We consider the lattice set $S=\mathbb{Z}^d$ with $d\ge 1$. The  \emph{state space} $(E, \mathscr{E},\lambda)$ is a measurable space equipped with a finite measure $\lambda$. Let $\Omega=E^{\mathbb{Z}^d}$ and ($\Omega, \scrf)$ be the \emph{configuration space}, where $\scrf$ is the $\sigma$-algebra generated by the cylinder sets. We denote by $\sigma=(\sigma_i)_{i\in \mathbb{Z}^d}\in \Omega$ a configuration in $\mathbb{Z}^d$, and for a subset $\Lambda \subset \mathbb{Z}^d$, we denote by $\sigma_{\Lambda}=(\sigma_i)_{i\in \Lambda}$ a configuration in $\Lambda$.
We use the notation $\Lambda \Subset \mathbb{Z}^d$ to denote that $\Lambda$ is a finite subset on $\mathbb{Z}^d$.
For $\Lambda \Subset \mathbb{Z}^d$, let $\scrf_{\Lambda}$ be the smallest $\sigma$-algebra on $\Omega$ containing $\{\sigma_{\Delta} \in A\}$ over all $\Delta \subset \Lambda, A \in \scre^{\Delta}.$

For each $x,y \in \mathbb{Z}^d$ with $x\neq y$, $\Phi_{\{x,y\}}: \Omega \to \mathbb{R}$ is a $\scrf_{\{x,y\}}$-measurable function. The collection $\Phi=\{\Phi_{\{x,y\}}\}_{x,y\in \mathbb{Z}^d}$ is called a \emph{potential}. We say that a potential is \emph{absolutely summable} if
\be\label{normphi}
\vertiii{\Phi} :=\sup_{x\in \mathbb{Z}^d}\sum_{\substack{y\in \mathbb{Z}^d \\ y\neq x}}\lVert \Phi_{\{x,y\}}\rVert<\infty,
\ee
where $\lVert \cdot \rVert$ denotes the sup-norm. We also assume that the potentials are translation-invariant.

Define the \emph{Hamiltonian} in a finite set $\Lambda$ with boundary condition $\omega \in \Omega$ associated to the potential $\Phi$ by 
\be
H_{\Lambda}^{\omega}(\sigma)= \sum_{\substack{x,y\in \Lambda \\ x\neq y }}\Phi_{\{x,y\}}(\sigma) + \sum_{\substack{x\in \Lambda \\ y\notin \Lambda }}\Phi_{\{x,y\}}(\sigma_{\Lambda}\omega_{\Lambda^c}),
\ee
where the configuration $((\sigma_{\Lambda}\omega_{\Lambda^c})_x)_{x\in \mathbb{Z}^d}$ means
\be
(\sigma_{\Lambda}\omega_{\Lambda^c})_x = 
\begin{cases}
\sigma_x &\text{ if } x\in \Lambda\\
\omega_x &\text{ if } x\notin \Lambda
\end{cases}.
\ee

For each $A \in \mathscr{F}_{\Lambda} $ and $\omega \in \Omega,$ we define the \emph{finite volume Gibbs measure} in $\Lambda$ with boundary condition $\omega$ and inverse temperature $\beta>0$ by 
\be
\mu^{\omega}_{\Lambda, \beta}\left(\sigma \right)=\frac{e^{-\beta H_{\Lambda}^{\omega}(\sigma)}}{Z_{\Lambda,\beta}^{\omega}},
\ee
where $Z_{\Lambda,\beta}^{\omega}$ is the \emph{partition function} given by
\be
Z_{\Lambda,\beta}^{\omega} = \int_{E^{\Lambda}}e^{-\beta H^{\omega}_{\Lambda}(\sigma)}\prod_{x \in \Lambda} \lambda(\diff\sigma_x).
\ee

\subsection{Local Central Limit Theorem}

We say that the random variable $X$ is \emph{lattice distributed} if there exist integer number $a$ and $h>0$ such that all possible values of $X$ may be represented in the form $a+bh$, where the parameter $b$ can assume any integer values. We call $h$ the {\em span of the distribution}. The distribution span $h$ is {\em maximal} if, no matter what the choice of $b\in \mathbb{Z}$ and $h_1>h$, it is impossible to represent all possible values of $X$ in the form $a+bh_1$. In this paper, we will always consider the span of the distribution to be maximal.

Consider $f:\Omega\to \mathbb{Z}$ be a $\scrf_0$-measurable function, where $0$ is the origin of the lattice $\mathbb{Z}^d$. By abuse of notation, we will frequently look $f$ as $f\circ \pi_0$, where $\pi_0:\Omega\to E$ is the projection function defined by $\pi_0(\sigma)=\sigma_0$.
Define
\be
\Var_{\lambda}(f)=\lambda\lf(f-\lambda(f))^2\ri.
\ee
If $f$ is bounded, i.e., $\lVert f\rVert<\infty$, then $\lambda(f)<\infty$ and $\Var_{\lambda}(f)<\infty$. Note that, if $\Var_{\lambda}(f)>0$, then $f$ is not $\lambda$-a.s. constant.
For each $x\in \mathbb{Z}^d$, define $\theta_x:\Omega \to \Omega$ be the {\em shift operator} $(\theta_x\sigma)_y = \sigma_{x+y}$, for all $y\in \mathbb{Z}^d$.

For a finite cube $\Lambda_k\Subset \mathbb{Z}^d$ given by $\Lambda_k=[-k,k]^d$, define, for a fixed $\sigma\in \Omega$,
\be
S_k(f) = \sum_{x\in \Lambda_k}f\circ \theta_x(\sigma) \quad \text{ and }\quad \bar{S}_k(f)=\frac{S_k(f)-\mu^\omega_{\Lambda_k,\beta}(S_k(f))}{\sqrt{D_k}},
\ee
where $D_k=D_k(f)=\mu^\omega_{\Lambda_k,\beta}((S_k(f)-\mu^\omega_{\Lambda_k,\beta}(S_k(f)))^2)$ denotes the variance of $S_k(f)$. For simplicity, we will use the notations $S_k=S_k(f)$ and $\bar{S}_k=\bar{S}_k(f)$.
 
For a sequence of increasing cubes $(\Lambda_k)_{k\ge 1}$ in $\mathbb{Z}^d$, and for a sequence of boundary conditions $(\omega_k)_{k\ge 1}$, 
we say that the $f$-\emph{integral central limit theorem} holds the the sequence $(\mu^{\omega_k}_{\Lambda_k,\beta})_{k\ge 1}$  if the following three conditions are satisfied:
\begin{enumerate}
\item[(i)] $\lim_{k\to \infty}D_k/\abs{\Lambda_k}=L$,
\item[(ii)] $L>0$,
\item[(iii)] For every $\tau\in \mathbb{R}$,
\be
\lim_{k\to \infty}\mu^{\omega_k}_{\Lambda_k,\beta}(\bar{S}_k\le \tau) = \frac{1}{\sqrt{2\pi}}\int_{-\infty}^{\tau} e^{-z^2/2}\diff z.
\ee
\end{enumerate}

Suppose that $f\circ \theta_x$ is a lattice distributed random variable for every $x\in \mathbb{Z}^d$ with the same $a\in \mathbb{Z}$ and maximal $h>0$. Throughout the paper, each time we use $f\circ \theta_x$ is lattice distributed in the statements of the results we mean that we pick the same $a$ and $h$ for every $x \in \mathbb{Z}^d.$ 
If $f$ is bounded, then the possible values for $f$ are $a+bh$ where $b$ is in the set $\{p,p+1,\ldots, q\}$ for some integers $p< q$. In this case, the possible values for $S_k$ are $\abs{\Lambda_k}a+bh$ where $b\in \mathfrak{B}_k:=\{p\abs{\Lambda_k},p\abs{\Lambda_k}+1,\ldots,q\abs{\Lambda_k}\}$. We also define $\mathfrak{B}:=\{p,\ldots,q\}$.

Define
\be 
z_{k,b}=\frac{\abs{\Lambda_k}a+bh-\mu^{\omega_k}_{\Lambda_k,\beta}(S_k)}{\sqrt{D_k}}.
\ee
We say that the $f$-\emph{local central limit theorem} holds for the sequence of Gibbs measures $(\mu^{\omega_k}_{\Lambda_k,\beta})_{k\ge 1}$ if the conditions (i) and (ii) are satisfied, and
\be
\lim_{k\to \infty}\sup_{b\in \mathfrak{B}_k} \abs{\frac{\sqrt{D_k}}{h}\mu^{\omega_k}_{\Lambda_k,\beta}(S_k=\abs{\Lambda_k}a+bh) -\frac{1}{\sqrt{2\pi}} \exp\left(-\frac{z_{k,b}^2}{2} \right) }=0.
\ee
The definition above comes from the version of the Local Limit Theorem for the lattice distributed random variables. For more details, see \cite{Gne}.

{\it Remark 1:} The Local Limit Theorem in Gnedenko's book (Chapter 8 in \cite{Gne}) for the i.i.d. lattice distributed random variables states that it is necessary and sufficient that the distribution span $h$ to be maximal. Note that, in our case, we will assume that $h$ is maximal.

Our main result is the following.

\begin{theorem}\label{main}
Suppose that $\Phi$ is a translation invariant and absolutely summable potential, and $f:\Omega \to \mathbb{Z}$ is a bounded $\scrf_0$-measurable function satisfying $\Var_{\lambda}(f)>0$ and $f\circ \theta_x$ is a lattice distributed random variable for every $x\in \mathbb{Z}^d$. Assume one of the following conditions,
\begin{enumerate}
\item\label{hightemp} the inverse temperature $\beta$ is sufficiently small.
\item\label{samecampanino} the potential $\Phi$ satisfies 
\be\label{campanino}
\sum_{\substack{x\in \mathbb{Z}^d \\ x\neq 0}}\lVert \Phi_{\{x,0\}} \lVert^{1/2} <\infty \quad \text{and}\quad \sup_{\lVert x\rVert \ge r}\lVert \Phi_{\{x,0\}}\rVert>0 \quad \text{for every }r>0.
\ee
\end{enumerate}
If the $f$-integral central limit theorem holds for a given sequence of Gibbs measures $(\mu^{\omega_k}_{\Lambda_k,\beta})_{k\ge 1}$, then the $f$-local central limit theorem holds for the same sequence of Gibbs measures.
\end{theorem}

The condition (\ref{samecampanino}) of Theorem \ref{main} is an extension of the result in \cite{CCT}, in the sense that we extend the result for a general two-body potentials and a general state space provided with a finite measure. Note that there are absolutely summable potentials for which the condition (\ref{campanino}) fails, for instance, the long-range Ising model (see Section \ref{application}). We are able to show that the result in \cite{CCT} is still true when we assume sufficiently high temperatures.

{\it Remark 2}: 
For the many-body interaction potential case, we believe that the absolute summability condition should be replaced with the condition in the respective norm for the convergence of the corresponding cluster expansion (see \cite{FV,Geo}).
 We expect that Theorem \ref{main} holds true at sufficiently high temperatures for the many-body interaction potential.
 
{\it Remark 3}: Dobrushin-Tirozzi \cite{DT} showed that if the short-range potential is bounded and translation-invariant, then given that CLT holds one can prove LLT as well. Note that their result holds for every temperature. We also expect that our result can be extended to any temperature.
 
 
  {\it Remark 4}: In the case when the model has a second order phase transition, we are uncertain to conclude Local Central Limit Theorem at the critical temperature. For instance, the 2D nearest-neighbor ferromagnetic Ising model has a second order phase transition (the result is true for every $d\ge 2$, see \cite{ADCS,AF,Yang}). However, at the critical temperature $T_c$, the susceptibility diverges. Thus, we cannot apply the result in \cite{New} to show the validity of the Central Limit Theorem. For the 2D Ising model with short-range interactions it is known \cite{Hil} that at the critical temperature, the CLT of another natural quantity that is the magnetization gives a non-Gaussian distribution in the limit compared with the off-critical situations where the CLT gives a Gaussian distribution.

As in \cite{CCT}, to obtain the main result, we use the following estimate
\begin{align}\label{diff}
&\sup_{b\in \mathfrak{B}_k} 2\pi\abs{\frac{\sqrt{D_k}}{h}\mu^{\omega_k}_{\Lambda_k,\beta}(S_k=\abs{\Lambda_k}a+bh) -\frac{1}{\sqrt{2\pi}} \exp\left(-\frac{z_{k,b}^2}{2} \right) }\\
&\leq \int_{-B}^{B}\abs{\mu^{\omega_k}_{\Lambda_k,\beta}\left(\exp \left(i t \bar{S}_{k}\right)\right)-\exp \left(-t^{2} / 2\right)} \diff t\nonumber \\
&\quad +\int_{\abs{t} \geq B} \exp \left(-t^{2} / 2\right) \diff t+\int_{B < \abs{t} < \delta \sqrt{D_{k}}} \abs{ \mu^{\omega_k}_{\Lambda_k,\beta}\left(\exp (i t \bar{S}_{k})\right)} \diff t\nonumber \\
&\quad+\int_{\delta \sqrt{D_{k}} \leq \abs{t} \leq \frac{\pi}{h} \sqrt{D_{k}}}\abs{\mu^{\omega_k}_{\Lambda_k,\beta}\left(\exp \left(i t \bar{S}_{k}\right)\right)} \diff t \nonumber.
\end{align}

If $B$ is large enough, the first integral is small by the Central Limit Theorem, and the second integral is also small. The third and fourth integrals follow from Lemma \ref{lem22} and \ref{lem23} for Theorem \ref{main} condition (\ref{hightemp}), and from Lemma \ref{lemcam1} and \ref{lemcam2} for condition (\ref{samecampanino}). All these lemmas will be proved in Section \ref{proofofthe mainresult}. The proofs are adaptation of Lemma 2.2 and 2.3 in \cite{CCT}. 

\subsection{Applications}\label{application}

Before the proof of Theorem \ref{main}, let us present some applications.

\subsubsection{Long-Range Ising Model}

Let $\Omega=\{-1,1\}^{\mathbb{Z}}$ be the set of configurations $\sigma=(\sigma_x)_{x\in \mathbb{Z}}$ on $\mathbb{Z}$. The Hamiltonian in a finite set $\Lambda$ with boundary condition $\omega$  is given by
\be \label{Hamil}
H^{\omega}_{\Lambda}(\sigma)=-\sum_{\substack{\{x,y\} \subset \Lambda \\ x\neq y}}J(\abs{x-y})\sigma_x\sigma_y - \sum_{\substack{x\in \Lambda \\ y\notin \Lambda}}J(\abs{x-y})\sigma_x\omega_y,
\ee
where the coupling constants $J_{x,y}=J(\abs{x-y})$, with $x\neq y$, are defined by
\be\label{coupling}
J(\abs{x-y})=
\begin{cases}
J &\text{ if }\abs{x-y}=1\\
\abs{x-y}^{-2+\alpha} &\text{ if }\abs{x-y}>1
\end{cases}
\ee
where $J(1)=J>0$ and $0\leq \alpha < 1$.
 It is known that these models undergo a phase transition at low temperatures \cite{Dys, FrSp}, they satisfy FKG inequality \cite{FKG}  and have finite susceptibility at high temperatures \cite{ACCN}.

As a first main application we obtain the Local Central Limit Theorem in the case of the long-range Ising model defined above. The fact that the Central Limit Theorem holds in the case of the long-range Ising model follows from \cite{New}.
In \cite{New}, it is proved that for models satisfying FKG inequality and that have finite susceptibility, the Central Limit Theorem holds for not necessarily monotonic functions of the random variables. This result extends the work in \cite{Newold} where the functions of the random variables are assumed to be monotonic. 

The Local Central Limit Theorem for these models at high temperatures is obtained as an application of our main result for $E=\{-1,1\}$ and $f(\sigma)=\sigma_0$. Moreover, note that the sum in (\ref{campanino}) diverges for this model.

\subsubsection{Gibbs Fields in a Compact Metric Space}

Under the Dobrushin's Uniqueness Condition \cite{Dob}, the work in \cite{Kun} proves that the Central Limit Theorem holds true for Gibbs fields with spins taking values in a compact metric space. In addition, in \cite{Kun} the assumption of the finite range of the potential previously considered in \cite{DT} is discarded. 

Thus, in our case as a second application, using our results we obtain that, at sufficiently high temperatures, the Local Central Limit Theorem holds true for the models considered in \cite{Kun}. These models extend the work in \cite{CCT} that treats only the case of discrete values for the spins, as well as the work in \cite{DT}, where a Local Central Limit Theorem is obtained under the assumption on the finite-range of the potential.

\section{Cluster Expansion at High Temperatures}\label{ce}

In this section, we are going to develop the cluster expansion at sufficiently high temperatures to control the absolute value of the characteristic function $\abs{\mu^{\omega_k}_{\Lambda_k,\beta}\left(\exp \left(i t \bar{S}_{k}\right)\right)}$ (see (\ref{diff})). The polymers and activity functions are similar with the ones in \cite{CCT}, and the proof of the absolutely convergence is an adaptation of the result in \cite{FV}, Chapters 5 and 6.

First, note that, for every $t\in \mathbb{R},$
\be\label{removebar}
\abs{\mu^{\omega_k}_{\Lambda_k,\beta}\left(\exp \left(i t \bar{S}_{k}\right)\right)}
= \abs{\mu^{\omega_k}_{\Lambda_k,\beta}\left(\exp \left(\frac{it}{{\sqrt{D_k}}} S_{k}\right)\right)}.
\ee
Define
\be
Z_{\Lambda_k, \beta,t}^{\omega} = \int_{E^{\Lambda_k}}e^{-\beta H^{\omega}_{\Lambda_k}(\sigma) +\frac{it}{{\sqrt{D_k}}} S_k}\prod_{x\in \Lambda_k} \lambda(\diff \sigma_x).
\ee
Note that $Z_{\Lambda_k,\beta,0}^{\omega}=Z_{\Lambda_k,\beta}^{\omega}$.

For all $x,y\in \mathbb{Z}^d$, since the function $\Phi_{\{x,y\}}$ is $\scrf_{\{x,y\}}$-measurable, by abuse of notation, we will start writing $\Phi_{\{x,y\}}(\sigma)=\Phi_{\{x,y\}}(\sigma_x,\sigma_y)$.

For a fixed $x\in \Lambda_k$, define 
\be
h_x^{\omega}(\sigma_x)=h_{x,\Lambda_k}^{\omega}(\sigma_x):=\sum_{y \notin \Lambda_k}\Phi_{\{x,y\}}(\sigma_x, \omega_y).
\ee 

For each $x\in \Lambda_k$ and $\omega \in \Omega$, define the probability density function $p^{\omega}_x:E\to \mathbb{R}$ by
\begin{equation}\label{probp}
p^{\omega}_x(\sigma_x) = p^{\omega}_{x,\Lambda_k,\beta}(\sigma_x) := \frac{\exp(-\beta h^{\omega}_x(\sigma_x))}{\displaystyle\int_{E} \exp(-\beta h^{\omega}_x(\eta_x))\lambda(\diff\eta_x)},
\end{equation}
and denote $\E^{\omega}_x$ be the expectation with respect to $p^{\omega}_x$.

Let $\mathcal{P}_{1,2}$ be a family of non-empty subsets $b\Subset \mathbb{Z}^d$, consisting of at most two points. A \emph{polymer} $R$ is a set $\{b_1,\ldots,b_p\}$ of elements $b_i \in \mathcal{P}_{1,2}$ that is connected in the following sense: for any $b_l,b_m \in R$, there exist $b_{k_1},\ldots, b_{k_q} \in R$ such that $b_{k_1}=b_l$, $b_{k_q}=b_m$ and $b_{k_j}\cap b_{k_{j+1}}\neq \emptyset$. Let $\mathcal{R}$ be the set of all polymers and, if $R\in \mathcal{R}$, denote by $\tilde{R}$ the subset of $\mathbb{Z}^d$ given by $\tilde{R} = \bigcup_{b\in R}b$. 

For $t\in \mathbb{R}$, define the activity function $\zeta^t_{\beta}: \mathcal{R}\to \mathbb{C}$ by
\begin{equation}\label{zetat}
\zeta^t_{\beta}(R):=\int_{E^{\tilde{R}}}\prod_{x \in \tilde{R}}p^{\omega}_x(\sigma_x) \prod_{b \in R}\xi_{\beta,b}(\sigma) \prod_{x\in \tilde{R}}\lambda(\diff \sigma_x),
\end{equation}
where we introduce the notations
\begin{align*}
\xi_{\beta,\{x\}}(\sigma) &= \xi^t_{\beta,\{x\}}(\sigma):=\exp\left( \frac{itf\circ \theta_x(\si)}{\sqrt{D_k}} \right)-1, \\ 
\xi_{\beta,\{x,y\}}(\sigma)&:=\exp(-\beta \Phi_{\{x,y\}}(\sigma_x, \sigma_y))-1.
\end{align*}

Let $\mathcal{P}_{i}$ be a family of non-empty subsets $b\Subset \mathbb{Z}^d$, consisting of $i$ points, and $\mathcal{P}_{i}(\Lambda_k)\subset \mathcal{P}_{i}$ be the subset of $ \mathcal{P}_{i}$ when $b\subset \Lambda_k$.
Define $\mathcal{R}_2$ to be the set of polymers $R\subset \mathcal{P}_2$. If $t=0$, we denote $\zeta_{\beta}(R):=\zeta^0_{\beta}(R)$ for every $R\in \mathcal{R}_2$.

For a polymer $R\in \mathcal{R}$, define $\gamma^i_R \subset R$ to be the set of all elements of $R$ with cardinality $i$. 
Note that $\tilde{\gamma}^1_R \cup\tilde{\gamma}^2_R=\tilde{R}$ and $\tilde{\gamma}^1_R \subset \tilde{\gamma}^2_R$ if $\abs{\tilde{R}}\ge 2$. Define, for connected $\gamma \subset \mathcal{P}_2(\Lambda_k)$,
\begin{equation}\label{hatzeta}
\hat{\zeta}_{\beta}(\gamma)=\int_{E^{\tilde{\gamma}}}\prod_{x \in \tilde{\gamma}}p^{\omega}_{x}(\sigma_x)\prod_{\{x, y\} \in \gamma}\abs{\xi_{\beta,\{x,y\}}(\sigma_x, \sigma_y)}\prod_{x \in \tilde{\gamma}}\lambda(\diff \sigma_x),
\end{equation}
and $\hat{\zeta}_{\beta}(\emptyset)=1$. For $n\ge 1$, define $\mathcal{L}^n\subset \mathcal{R}^n$ be the set of ordered $n$-tuples $(R_1,\ldots,R_n)$ such that there exists $1\le i\le n $ satisfying $\abs{\tilde{R}_i}>1$. 
 
For an ordered $n$-tuple $(R_1,\ldots,R_n)\in \mathcal{R}^n$, define the \emph{Ursell function} 
$$
\phi^T(R_1,\ldots, R_n) =
\begin{cases}
1 &\text{ if } n=1\\
\displaystyle\sum_{\substack{G \subset \mathcal{G}_{\{R_1,\ldots,R_n\}} \\ G \ \text{conn. spann.}}}\frac{(-1)^{e(G)}}{n!} &\text{ if } n\ge 2,\ \mathcal{G}_{\{R_1,\ldots,R_n\}}\text{ conn.}\\
0 &\text{ if } n\ge 2,\ \mathcal{G}_{\{R_1,\ldots,R_n\}}\text{ not conn.}
\end{cases}
$$
where $\mathcal{G}_{\{R_1,\ldots,R_n\}}$ is the graph of vertices $\{1,\ldots,n \}$ and edges
\be\label{edges}
\{\{i,j\}: \tilde{R}_i\cap \tilde{R}_j \neq \emptyset, 1\le i,j\le n, i\neq j\},
\ee
and $G$ ranges over all its connected spanning subgraphs. We denote $e(G)$ be the number of edges in $G$.

\begin{proposition}\label{theoremlemma1}
Suppose that $\Phi$ is a translation invariant and absolutely summable potential, and $f:\Omega \to \mathbb{Z}$ is a bounded $\scrf_0$-measurable function satisfying $\Var_{\lambda}(f)>0$. For a fixed $C\in (0,e^{-1})$ and $\beta>0$, there exist $a_\beta=a(\beta,C)$ and $\beta_C>0$ such that, for every $\beta<\beta_C$,
\be\label{etathm2}
\sum_{n=1}^{\infty} \sum_{(R_1,\ldots,R_n)\in \mathcal{L}^n} \abs{\phi^T(R_1,\ldots,R_n)} \prod_{i=1}^n C^{\abs{\tilde{\gamma}^1_{R_i}}}\hat{\zeta}_{\beta}(\gamma^2_{R_i}) \le a_{\beta} \abs{\Lambda_k}.
\ee
\end{proposition}

\begin{proof}
For a fixed polymer $R_0\in \mathcal{R}$,
\be\label{unir}
\sum_{ x\in \tilde{R}_0}C^{\abs{\tilde{\gamma}^1_{\{x\}}}}\hat{\zeta}_{\beta}(\gamma^2_{\{x\}})= C \abs{\tilde{R}_0}
\ee
and
\begin{align}
\sum_{\substack{R: \tilde{R} \cap \tilde{R}_0 \neq \emptyset \\ \abs{\tilde{R}}\ge 2}}C^{\abs{\tilde{\gamma}^1_R}}\hat{\zeta}_{\beta}(\gamma^2_R)e^{\abs{\tilde{\gamma}^2_R}}
&=\sum_{\gamma^2 : \tilde{\gamma}^2 \cap \tilde{R}_0\neq \emptyset} \hat{\zeta}_{\beta}(\gamma^2)e^{\abs{\tilde{\gamma}^2}}\sum_{\gamma^1: \tilde{\gamma}^1 \subset \tilde{\gamma}^2}C^{\abs{\tilde{\gamma}^1}}\nonumber\\
&= \sum_{\gamma^2: \tilde{\gamma}^2\cap\tilde{\gamma}^2_{R_0} \neq \emptyset}[(1+C) e]^{\abs{\tilde{\gamma}^2}}\hat{\zeta}_{\beta}(\gamma^2). \label{ineqhat}
\end{align}
Define
\begin{equation}
a_{\beta}:=(1+C)^2 e^{4}\sup_{x \in \mathbb{Z}^d}\sum_{y \neq x}\lVert e^{-\beta \Phi_{\{x, y\}}}-1\rVert.
\end{equation}
Since the potential $\Phi$ is absolutely summable and $Ce<1$, there exists $\beta_C >0$ such that, for all $\beta <\beta_C$, we have $Ce+a_{\beta} < 1$.

By the same argument as in \cite{FV} (Lemma 6.99), for every $\beta<\beta_C$,
\begin{equation}
\max_{z \in \Lambda_k}\sum_{\gamma^2: z \in \tilde{\gamma}^2}[(1+C) e]^{\abs{\tilde{\gamma}^2}}\hat{\zeta}_{\beta}(\gamma^2)\le a_{\beta}.
\end{equation}

Thus,
\be
\sum_{R: \tilde{R} \cap \tilde{R}_0 \neq \emptyset}C^{\abs{\tilde{\gamma}^1_R}}\hat{\zeta}_{\beta}(\gamma^2_R)e^{\abs{\tilde{R}}} \le (Ce+a_{\beta})\abs{\tilde{R}_0}< \abs{\tilde{R}_0}.
\ee
Applying Theorem 5.4 in \cite{FV}, we have
\be
1+\sum_{n=2}^{\infty}n\sum_{(R_2,\ldots, R_n)}\abs{\phi^T(R_1,\ldots,R_n)} \prod_{i=2}^n C^{\abs{\tilde{\gamma}^1_{R_i}}}\hat{\zeta}_{\beta}(\gamma^2_{R_i})\le e^{\abs{\tilde{R}_1}},
\ee
where the second sum is over all ordered $(n-1)$-tuple $(R_2,\ldots,R_n)\in \mathcal{R}^{n-1}$ of polymers such that $\tilde{R}_i\subset \Lambda_k$ for all $2\le i\le n$.
Therefore,
\begin{align*}
&\sum_{n=1}^{\infty} \sum_{(R_1,\ldots,R_n)\in \mathcal{L}^n} \abs{\phi^T(R_1,\ldots,R_n)} \prod_{i=1}^n C^{\abs{\tilde{\gamma}^1_{R_i}}}\hat{\zeta}_{\beta}(\gamma^2_{R_i})\\
&=\sum_{\substack{R_1\in\mathcal{R} \\ \abs{\tilde{R}_1}\ge 2}}\!\! C^{\abs{\tilde{\gamma}^1_{R_1}}}\hat{\zeta}_{\beta}(\gamma^2_{R_1})\lf 1+\sum_{n=2}^{\infty}n \!\! \sum_{(R_2,\ldots,R_n)}\abs{\phi^T(R_1,\ldots,R_n)} \prod_{i=2}^n C^{\abs{\tilde{\gamma}^1_{R_i}}}\hat{\zeta}_{\beta}(\gamma^2_{R_i})  \ri\\
&\le a_{\beta}\abs{\Lambda_k},
\end{align*}
as we desired.
\end{proof}

The next theorem shows that the partition function $Z^{\omega}_{\Lambda_k,\beta,t}$ can be written as a polymer partition function $\Xi^{\omega}_{\Lambda_k, \beta,t}$ with activity function $\zeta^t_{\beta}$. Moreover, for $\abs{t}<\delta\sqrt{D_k}$, the function $\log \Xi^{\omega}_{\Lambda_k, \beta,t}$ can  be expressed as an absolutely convergent series when the temperature is sufficiently large.

\begin{theorem}\label{clusterexpansionthm}
Suppose that $\Phi$ is a translation invariant and absolutely summable potential, and $f:\Omega \to \mathbb{Z}$ is a bounded $\scrf_0$-measurable function satisfying $\Var_{\lambda}(f)>0$.
For every $k\ge 1$, $t\in \mathbb{R}$, $\beta>0$, and $w\in \Omega$, the partition function $Z^{\omega}_{\Lambda_k,\beta,t}$ can be written as
\begin{equation}\label{campanino2}
Z_{\Lambda_k, \beta,t}^{\omega}=\left(\prod_{x \in \Lambda_k}\int_{E} e^{-\beta h_x^{\omega}(\sigma_x)} \lambda(\diff \si_x)\right)\Xi^{\omega}_{\Lambda_k, \beta,t},
\end{equation}
where
\begin{equation}
\Xi^{\omega}_{\Lambda_k, \beta,t}:=1+\sum_{n=1}^\infty \sum_{\substack{\{R_1,\ldots, R_n\}\\ \tilde{R}_i \cap \tilde{R}_j =\emptyset,i \neq j}}\prod_{i=1}^n\zeta^t_{\beta}(R_i).
\end{equation}



Moreover, there exists $\delta_0>0$ such that, for every $0<\delta<\delta_0$, there exists $\beta_{\delta}>0$ such that, for every $\beta<\beta_{\delta}$, if $0<\abs{t}<\delta\sqrt{D_k}$, there exists $\alpha_{\delta,\beta}<1$
satisfying
\be\label{velenik2}
\sum_{n=1}^{\infty} \sum_{(R_1,\ldots,R_n)\in \mathcal{R}^n} \abs{\phi^T(R_1,\ldots,R_n)} \prod_{i=1}^n \abs{\zeta^t_{\beta}(R_i)} \le \alpha_{\delta,\beta}\abs{\Lambda_k}.
\ee
For $t=0$, there exists $\beta_0>0$ such that, for every $\beta<\beta_0$, there exists $\alpha_{\beta}<1$ satisfying
 \be\label{tzero}
\sum_{n=1}^{\infty} \sum_{(R_1,\ldots,R_n)\in \mathcal{R}^n_2} \abs{\phi^T(R_1,\ldots,R_n)} \prod_{i=1}^n \abs{\zeta_{\beta}(R_i)} \le \alpha_{\beta}\abs{\Lambda_k}.
 \ee
 
Moreover, both $\alpha_{\delta,\beta}$ and $\alpha_{\beta}$ decrease to 0 when $\delta\to 0$ and $\beta\to 0$.

Under the above conditions, for every $\abs{t}<\delta \sqrt{D_k}$, the function $\Xi^{\omega}_{\Lambda_k,\beta,t}$ can be expressed as the following,
\be\label{cluster}
\Xi^{\omega}_{\Lambda_k,\beta,t} = \exp\left( \sum_{n=1}^{\infty} \sum_{(R_1,\ldots,R_n)} \phi^T(R_1,\ldots,R_n) \prod_{i=1}^n \zeta^t_{\beta}(R_i) \right),
\ee
where the second sum in (\ref{cluster}) is over $\mathcal{R}^n$ for $0<\abs{t}<\delta\sqrt{D_k}$, and over $\mathcal{R}^n_2$ for $t=0$.
\end{theorem}

\begin{proof}

By Fubini's Theorem, let us write the partition function $Z^{\omega}_{\Lambda_k,\beta,t}$ depending on $\Xi^{\omega}_{\Lambda_k, \beta,t}$,
\begin{align*}
&\int_{E^{\Lambda_k}}\prod_{x\in \Lambda_k}e^{it\frac{f\circ \theta_x(\sigma)}{\sqrt{D_k}}}\prod_{\{x,y\}\in \mathcal{P}_2(\Lambda_k)}e^{-\beta \Phi_{\{x,y\}}(\sigma_x,\sigma_y)}\prod_{x\in \Lambda_k}e^{-\beta h^{\omega}_x(\sigma_x)}\prod_{x\in \Lambda_k} \lambda(\diff \sigma_x)\\
=&\sum_{n=0}^{\infty}\sum_{\substack{\{R_1,\ldots,R_n\} \\ \tilde{R}_i\cap \tilde{R}_j=\emptyset, i\neq j}}\int_{E^{\Lambda_k}}\prod_{i=1}^n\prod_{b\in R_i}\xi_{\beta,b}(\sigma)\prod_{x\in \Lambda_k}e^{-\beta h^{\omega}_x(\sigma_x)}\prod_{x\in \Lambda_k} \lambda(\diff \sigma_x)\\
=&\sum_{n=0}^{\infty}\sum_{\substack{\{R_1,\ldots,R_n\} \\ \tilde{R}_i\cap \tilde{R}_j=\emptyset, i\neq j}} \Bigg[ \lf \int_{E^{\Lambda_k\setminus\cup_{i=1}^n \tilde{R}_i}}\prod_{x\in \Lambda_k\setminus\cup_{i=1}^n \tilde{R}_i}e^{-\beta h^{\omega}_x(\sigma_x)}
\prod_{x\in \Lambda_k\setminus\cup_{i=1}^n \tilde{R}_i}\lambda(\diff \sigma_x) \ri\\
&\ \cdot
\prod_{i=1}^n \int_{E^{\tilde{R}_i}}\prod_{b\in R_i}\xi_{\beta,b}(\sigma)\prod_{x\in \tilde{R}_i}e^{-\beta h^{\omega}_x(\sigma_x)}\prod_{x\in \tilde{R}_i}\lambda(\diff \sigma_x)\Bigg]\\
=&\sum_{n=0}^{\infty}\sum_{\substack{\{R_1,\ldots,R_n\} \\ \tilde{R}_i\cap \tilde{R}_j=\emptyset, i\neq j}} \Bigg[\lf \prod_{x\in \Lambda_k\setminus\cup_{i=1}^n \tilde{R}_i}\int_{E}e^{-\beta h^{\omega}_x(\sigma_x)} \lambda(\diff \sigma_x) \ri\\
&\ \cdot
\prod_{i=1}^n \int_{E^{\tilde{R}_i}}\prod_{b\in R_i}\xi_{\beta,b}(\sigma)\prod_{x\in \tilde{R}_i}e^{-\beta h^{\omega}_x(\sigma_x)}\prod_{x\in \tilde{R}_i}\lambda(\diff \sigma_x)\Bigg]\\
=& \lf \prod_{x\in \Lambda_k}\int_{E}e^{-\beta h^{\omega}_x(\sigma_x)} \lambda(\diff \sigma_x) \ri \lf 1+ \sum_{n=1}^{\infty}\sum_{\substack{\{R_1,\ldots,R_n\} \\ \tilde{R}_i\cap \tilde{R}_j=\emptyset, i\neq j}} \prod_{i=1}^n \zeta^t_{\beta}(R_i)\ri.
\end{align*}

For $t=0$, the proof of Equation (\ref{tzero}) follows by a very similar argument as in \cite{FV} (Lemma 6.99). Assume $0<\abs{t}<\delta\sqrt{D_k}$ for some $\delta>0$ (we choose a suitable $\delta$ along the proof).

Note that for every $x \in \mathbb{Z}^d$, since $\abs{f\circ \theta_x(\sigma)} \leq \lVert f \rVert$, there exists $\delta_1>0$ such that, for every $0<\delta<\delta_1$,
\begin{equation}\label{ineqxi}
\abs{\xi_{\beta,\{x\}}(\sigma)}=\sqrt{2-2\cos\left(\frac{t}{\sqrt{D_k}}f\circ \theta_x(\sigma)\right)}\le \delta \lVert f\rVert.
\end{equation}
Thus, for every polymer $R$ with $\tilde{R} \subseteq \Lambda_k$,
\be
\abs{\zeta^t_{\beta}(R)} \leq (\delta \lVert f\rVert)^{\abs{\tilde{\gamma}^1_R}}\hat{\zeta}_{\beta}(\gamma^2_R).
\ee
In particular, $\abs{\zeta^t_{\beta}(\{x\})}\le \delta \lVert f\rVert$ for every $x\in \Lambda_k$. Note that
\be\label{phiy}
\phi^T(\underbrace{\{x\},\ldots,\{x\}}_{n \text{ copies}})=(-1)^{n-1}\frac{(n-1)!}{n!}=\frac{(-1)^{n-1}}{n}.
\ee
Thus,
\be
\sum_{n=1}^{\infty}\sum_{x\in \Lambda_k}\abs{\phi^T(\{x\},\ldots,\{x\})}\abs{\zeta^t_{\beta}(\{x\})^n} \le A(\delta)\abs{\Lambda_k},
\ee
where, for every $\delta<\lVert f\rVert^{-1}$,
\be
A(\delta):= \sum_{n=1}^{\infty}\frac{1}{n}(\delta \lVert f\rVert)^n<1.
\ee

Choosing $C=\delta \lVert f\rVert$, there exists $\delta_0<\min\{\delta_1,\lVert f\rVert^{-1}\}$ such that, for every $\delta<\delta_0$, we have $Ce<1$. By Proposition \ref{theoremlemma1}, there exists $\beta_{\delta}>0$ such that, for every $\beta<\beta_{\delta}$,
\begin{align*}
&\sum_{n=1}^{\infty} \sum_{(R_1,\ldots,R_n)} \abs{\phi^T(R_1,\ldots,R_n)} \prod_{i=1}^n \abs{\zeta^t_{\beta}(R_i)} \\
&\le A(\delta)\abs{\Lambda_k}+\sum_{n=1}^{\infty} \sum_{(R_1,\ldots,R_n)\in \mathcal{L}^n} \abs{\phi^T(R_1,\ldots,R_n)} \prod_{i=1}^n (\delta \lVert f\rVert)^{\abs{\tilde{\gamma}^1_{R_1}}}\hat{\zeta}_{\beta}(\gamma^2_{R_i})\\
&\le (A(\delta) + a_{\beta})\abs{\Lambda_k}.
\end{align*}

Therefore, we have (\ref{velenik2}) choosing $\alpha_{\delta,\beta}=A(\delta) + a_{\beta}$.
The proof of Expansion (\ref{cluster}) will be omitted since it is similar to the argument in \cite{FV} (Chapter 5).
\end{proof}

For a fixed $c>0$ and $\beta>0$, define the activity function $\eta^c_{\beta}:\mathcal{R}_2\to \mathbb{R}$ by 
\be
\eta^c_{\beta}(R)=e^{c\abs{\tilde{R}}}\hat{\zeta}_{\beta}(R).
\ee
The next theorem shows a condition of a convergence of a polymer cluster expansion $\Xi^{\omega}_{\Lambda_k, \beta}(\eta^c_{\beta})$ with  activity function $\eta^c_{\beta}$. This will be useful to prove Lemma \ref{lem23} in Section \ref{proofofthe mainresult}, where we are interested in the region $\delta\sqrt{D_k}\le \abs{t}\le \pi\sqrt{D_k}/h$ at sufficiently large temperatures.

\begin{theorem}\label{theoremeta}
Suppose that $\Phi$ is a translation invariant and absolutely summable potential, and $f:\Omega \to \mathbb{Z}$ is a bounded $\scrf_0$-measurable function satisfying $\Var_{\lambda}(f)>0$.
For every $c>0$ there exists $\beta_c>0$ such that, for every $\beta<\beta_c$, there exists $\bar{\alpha}_{c,\beta}<1$ satisfying
\be\label{etathm}
\sum_{n=1}^{\infty} \sum_{(R_1,\ldots,R_n)\in \mathcal{R}^n_2} \abs{\phi^T(R_1,\ldots,R_n)} \prod_{i=1}^n \eta^c_{\beta}(R_i) \le \bar{\alpha}_{c,\beta} \abs{\Lambda_k}.
\ee
Moreover, under the above conditions, the polymer partition function
\begin{equation}\label{xiforeta}
\Xi^{\omega}_{\Lambda_k, \beta}(\eta^c_{\beta}):=1+\sum_{n=1}^\infty \sum_{\substack{\{R_1,\ldots, R_n\}\\ \tilde{R}_i \cap \tilde{R}_j =\emptyset,i \neq j}}\prod_{i=1}^n\eta^c_{\beta}(R_i)
\end{equation}
can be written as
\be\label{xietathm}
\Xi^{\omega}_{\Lambda_k, \beta}(\eta^c_{\beta})=\exp\left( \sum_{n=1}^{\infty} \sum_{(R_1,\ldots,R_n)\in \mathcal{R}^n_2} \phi^T(R_1,\ldots,R_n) \prod_{i=1}^n \eta^c_{\beta}(R_i) \right).
\ee
\end{theorem}

\begin{proof}
For a fixed $c>0$, define
\be
\bar{\alpha}_{c,\beta}:=e^{2(2+c)}\sup_{x\in \mathbb{Z}^d}\sum_{y\neq x}\lVert e^{-\beta \Phi_{\{x,y\}}}-1 \rVert.
\ee
Since $\Phi$ is absolutely summable, there exists $\beta_c>0$ such that, for every $\beta<\beta_c$, we have $\bar{\alpha}_{c,\beta}<1$. Moreover, $\bar{\alpha}_{c,\beta}\to 0$ when $\beta\to 0$. 
For a fixed polymer $R_1\in \mathcal{R}_2$, following the proof of Lemma 6.99 in \cite{FV}, we have
\be
\sum_{\substack{R\in \mathcal{R}_2 \\ \tilde{R}\cap \tilde{R}_1\neq \emptyset}}\eta^c_{\beta}(R)e^{\abs{\tilde{R}}}
= \sum_{\substack{R\in \mathcal{R}_2\\ \tilde{R}\cap \tilde{R}_1\neq \emptyset}}\hat{\zeta}_{\beta}(R)e^{(1+c)\abs{\tilde{R}}}\le
\bar{\alpha}_{c,\beta} \abs{\tilde{R}_1}.
\ee
The proof for Estimate (\ref{etathm}) finishes applying Theorem 5.4 in \cite{FV}. 
\end{proof}

\section{Proof of the main result}\label{proofofthe mainresult}

In this section we prove the main result. The section is divided into two subsections in which we prove each condition of Theorem \ref{main}. We start the proof with Proposition \ref{prop1} that is an adaptation of Proposition 3.3 in \cite{CCT}. To show it, we need the following lemma that can be found in \cite{DT} and in \cite{Gne}.

\begin{lemma}\label{lemmadt}
If $X$ is a lattice distributed random variable with maximal span of the distribution $h>0$, then for every $\varepsilon>0$, it is possible to find a positive constant $d_X$ such that for every $t$, $\varepsilon\le \abs{t} \le \frac{2\pi}{h}-\varepsilon$ it is true the following inequality,
\be
\abs{\mathbb{E}(e^{itX})}\le e^{-d_X}.
\ee
\end{lemma}

\begin{proposition}\label{prop1}
Suppose that $\Phi$ is a translation invariant and absolutely summable potential, and $f:\Omega \to \mathbb{Z}$ is a bounded $\scrf_0$-measurable function satisfying $\Var_{\lambda}(f)>0$.
\begin{enumerate}
\item[(a)] For every $\beta>0$, there exists a positive $d(\beta)$ such that, for every $\Lambda\Subset \mathbb{Z}^d$, and $x\in \Lambda$,
\be\label{variancepx}
\E^{\omega}_x((f\circ \theta_x(\sigma))^2)\ge d(\beta)
\ee
uniformly with respect to $\omega$. Moreover, $d(\beta)$ is decreasing in $\beta$.
\item[(b)] Assume that $f\circ \theta_x$ is a lattice distributed random variable for every $x\in \mathbb{Z}^d$. For every $\beta>0$ and $0<\delta<\pi/h$, there exists a positive constant $c=c(\beta,\delta)$ such that, for every $\Lambda\Subset \mathbb{Z}^d$, and $x\in \Lambda$,
\be\label{momentpx}
\abs{\E^{\omega}_x(e^{itf\circ \theta_x(\sigma)})}<e^{-c},
\ee
for all $\delta\le \abs{t}\le \pi/h$ uniformly with respect to $\omega$.
\item[(c)] Assume that $f\circ \theta_x$ is a lattice distributed random variable for every $x\in \mathbb{Z}^d$. For every $0<\delta<\pi/h$, there exist $\beta'_{\delta}>0$ and a positive constant $c=c(\beta'_{\delta},\delta)$ such that, for every $\Lambda\Subset \mathbb{Z}^d$, and $x\in \Lambda$,
\be\label{momentpx2}
\abs{\E^{\omega}_x(e^{itf\circ \theta_x(\sigma)})}<e^{-c},
\ee
for all $\delta\le \abs{t}\le \pi/h$ and $\beta<\beta'_{\delta}$ uniformly with respect to $\omega$.
\end{enumerate}
\end{proposition}

\begin{proof}
Estimate (\ref{variancepx}) follows from
\be
p^{\omega}_x(\sigma_x) \ge \frac{e^{-2\beta \vertiii{\Phi}}}{\lambda(E)}.
\ee
Thus,
\be
\E^{\omega}_x((f\circ \theta_x(\sigma))^2)\ge \frac{e^{-2\beta \vertiii{\Phi}}}{\lambda(E)}\lambda(f^2) :=d(\beta).
\ee

Note that
\be
\E^{\omega}_x(e^{itf\circ \theta_x(\sigma)})=\sum_{n=-\infty}^{\infty}e^{itn}\mathbb{P}^{\omega}_x(f=n),
\ee
where $\mathbb{P}^{\omega}_x(A)=\int_A p^{\omega}_x(\sigma_x)\lambda(\diff \sigma_x)$ for all $A\in \mathscr{E}$. Then (\ref{momentpx}) is obtained from Lemma \ref{lemmadt}.

To prove (\ref{momentpx2}), define the probability measure $\nu(A)=\lambda(A)/\lambda(E)$ for all $A\in \mathcal{E}$. By Lemma \ref{lemmadt}, there exists $c_0>0$ such that $\abs{\nu\lf e^{itf\circ \theta_x(\sigma)}\ri}<e^{-c_0}$ for all $\delta\le \abs{t}\le \pi/h$. 
Define the set
\be
\mathfrak{A}=\{a+bh\in \mathbb{Z}: b\in \mathfrak{B}\}.
\ee

Choose $\varepsilon>0$ such that $\varepsilon + e^{-c_0}<1$.
Since
\be
e^{-2\beta \vertiii{\Phi}}\nu(A)\le \mathbb{P}^{\omega}_x(A) \le e^{2\beta \vertiii{\Phi}}\nu(A)
\ee
for every $A\in \mathcal{E}$, and $f$ is bounded, there exists $\beta'_{\delta}>0$ such that, for every $\beta<\beta'_{\delta}$,
\be
\abs{\mathbb{P}^{\omega}_x(f=n) - \nu(f=n)}<\frac{\varepsilon}{2\lVert f\rVert +1}
\ee
for every $n\in \mathfrak{A}$. Thus,
\begin{align*}
\abs{\sum_{n=-\infty}^{\infty}e^{itn}\mathbb{P}^{\omega}_x(f=n)} &\le 
\sum_{n\in \mathfrak{A}}\abs{\mathbb{P}^{\omega}_x(f=n) -\nu(f=n) } + \abs{\nu\lf e^{itf\circ \theta_x(\sigma)}\ri}\\
&\le \varepsilon+e^{-c_0}\\
&=e^{-c}
\end{align*}
for some $c=c(\beta'_{\delta},\delta)>0$.
\end{proof}

\subsection{Proof of Theorem \ref{main} Condition (\ref{hightemp})}

\begin{lemma}\label{lem22}
Suppose that $\Phi$ is a translation invariant and absolutely summable potential, and $f:\Omega \to \mathbb{Z}$ is a bounded $\scrf_0$-measurable function satisfying $\Var_{\lambda}(f)>0$. There exists $\delta>0$ and $\beta(\delta)>0$, and a positive constant $D=D(\delta,\beta(\delta))$, not depending on $\omega_k$ and $\Lambda_k$, such that, for every $\beta<\beta(\delta)$ and $\abs{t}< \delta\sqrt{D_k}$,
\be
\abs{\mu^{\omega_k}_{\Lambda_k,\beta}\left(\exp \left(i t \bar{S}_k\right) \right)} \leq \exp \left(-t^{2}D\frac{\abs{\Lambda_{k}}}{D_{k}}\right).
\ee
\end{lemma}

\begin{proof}

By Theorem \ref{clusterexpansionthm}, there exists $\delta_0>0$ such that, for every $0<\delta<\delta_0$, choose $\beta_{\delta}>0$ and $\beta_0>0$ such that Equation (\ref{cluster}) holds for every $\beta<\min\{\beta_{\delta},\beta_0\}$ and $\abs{t}<\delta\sqrt{D_k}$. Using Equation (\ref{removebar}),  by Taylor Remainder Theorem, there exists $0<\theta< \delta \sqrt{D_k}$ such that 

\begin{align*}
 &\abs{\mu^{\omega_k}_{\Lambda_k,\beta}\left(\exp \left(i t \bar{S}_k\right) \right)}=\abs{ \mu^{\omega_k}_{\Lambda_k,\beta} \left(\exp \left(\frac{i t}{ \sqrt{D_{k}}} \sum_{x \in \Lambda_{k}}f\circ \theta_x(\sigma)\right) \right) }\nonumber \\
&\leq \exp \left(\frac{t^{2}}{2} \sum_{y\in \Lambda_k}  \operatorname{Re} \frac{d^{2}}{d t^{2}} \zeta^t_{\beta}(\{y\})\Bigg\vert_{t=\theta} \right)
\cdot \exp \lf \frac{t^2}{4}\sum_{y\in \Lambda_k}\operatorname{Re}\frac{d^{2}}{d t^{2}} \zeta^t_{\beta}(\{y\})^2\Bigg\vert_{t=\theta} \ri\\
&\quad \cdot \exp\lf \frac{t^2}{2} \sum_{n=3}^{\infty}\sum_{y\in \Lambda_k} \abs{\phi^T(\{y\},\ldots,\{y\})}\abs{\frac{d^{2}}{d t^{2}} \zeta^t_{\beta}(\{y\})^n\Bigg\vert_{t=\theta}  } \ri\\
&\quad \cdot \exp \left(\frac{t^{2}}{2}\sum_{n=1}^{\infty}\sum_{(R_{1}, \ldots, R_{n})\in \mathcal{L}^n}\abs{\phi^{T}\left(R_{1}, \ldots, R_{n}\right)}\abs{\frac{d^{2}}{d t^{2}} \prod_{i=1}^{n} \zeta^t_{\beta}\left(R_{i}\right)\Bigg\vert_{t=\theta} }\right).
\end{align*}

Differentiating the activity functions, there exists $\delta_1<\delta_0$ such that, for every $\delta<\delta_1$, for each $y\in \Lambda_k$, we have
\be
\operatorname{Re}\lf \lf\frac{d}{dt}\zeta^t_{\beta}(\{y\}) \ri^2 \ri \le -\frac{\cos(2\delta \lVert f\rVert)}{D_k}\lf \mathbb{E}_y^{\omega}(f\circ \theta_y)\ri^2\le 0
\ee
and
\be
\abs{\zeta^t_{\beta}(\{y\})}\le \delta\lVert f\rVert, \quad \abs{\frac{d}{dt}\zeta^t_{\beta}(\{y\})}\le \frac{\lVert f\rVert}{\sqrt{D_k}}, \quad \abs{\frac{d^2}{dt^2}\zeta^t_{\beta}(\{y\})}\le \frac{\lVert f\rVert^2}{D_k}.
\ee
These bounds imply
\be
\operatorname{Re}\frac{d^{2}}{d t^{2}} \zeta^t_{\beta}(\{y\})^2\Bigg\vert_{t=\theta} \le 2\delta\frac{\lVert f \rVert^3}{D_k}
\ee
and, for all $n\ge 3$,
\be
\abs{\frac{d^{2}}{d t^{2}} \zeta^t_{\beta}(\{y\})^n\Bigg\vert_{t=\theta}} \le \frac{\lVert f\rVert^2}{D_k} n\lf (n-1)(\delta\lVert f\rVert)^{n-2}+(\delta \lVert f\rVert)^{n-1}\ri.
\ee
Using Equation (\ref{phiy}), the series
\be
B(\delta):=\sum_{n=3}^{\infty}\lf (n-1)(\delta\lVert f\rVert)^{n-2}+(\delta \lVert f\rVert)^{n-1}\ri
\ee
is convergent for $\delta<\lVert f\rVert^{-1}$ and $B(\delta)\to 0$ when $\delta\to 0$.

Let us recall the function $\hat{\zeta}_{\beta}$ in (\ref{hatzeta}). There exists $\delta_2<\min\{\delta_1,\lVert f\rVert^{-1}\}$ such that, for every $\delta<\delta_2$ and for a every polymer $R\in \mathcal{R}$, since $\abs{\xi^t_{\{y\},\beta}(\sigma)}\le \delta \lVert f \rVert$ for every $\{y\}\in \gamma^1_R$ and $0<\abs{t}<\delta\sqrt{D_k}$, we have
\be\label{der0}
\abs{\zeta^t_{\beta}(R)}\le (\delta \lVert f \rVert)^{\abs{\tilde{\gamma}^1_R}}\hat{\zeta}_{\beta}(\gamma^2_R).
\ee
Since
\begin{align*}
\abs{\frac{d}{dt}\prod_{\{y\}\in \gamma^1_R}\xi^t_{\{y\},\beta}(\sigma_y)}
&\le \frac{\lVert f\rVert}{\sqrt{D_k}}\abs{\tilde{\gamma}^1_R}(\delta\lVert f\rVert)^{\abs{\tilde{\gamma}^1_R}-1},\\
\abs{\frac{d^2}{dt^2}\prod_{\{y\}\in \gamma^1_R}\xi^t_{\{y\},\beta}(\sigma_y)}
&\le
\frac{\lVert f\rVert^2}{D_k}\abs{\tilde{\gamma}^1_R}\lf (\abs{\tilde{\gamma}^1_R}-1)(\delta\lVert f\rVert)^{\abs{\tilde{\gamma}^1_R}-2}+(\delta\lVert f\rVert)^{\abs{\tilde{\gamma}^1_R}-1}\ri,
\end{align*}
we have, respectively, the following bounds,
\begin{align*}
\abs{\frac{d}{dt}\zeta^t_{\beta}(R)}
&\le \frac{\lVert f\rVert}{\sqrt{D_k}} \abs{\tilde{\gamma}^1_R}(\delta\lVert f\rVert)^{\abs{\tilde{\gamma}^1_R}-1} \hat{\zeta}_{\beta}(\gamma^2_R),\\
\abs{\frac{d^2}{dt^2}\zeta^t_{\beta}(R)}
&\le \frac{\lVert f\rVert^2}{D_k}\abs{\tilde{\gamma}^1_R}\lf (\abs{\tilde{\gamma}^1_R}-1)(\delta\lVert f\rVert)^{\abs{\tilde{\gamma}^1_R}-2}+(\delta\lVert f\rVert)^{\abs{\tilde{\gamma}^1_R}-1}\ri\hat{\zeta}_{\beta}(\gamma^2_R).
\end{align*}

Therefore, using
\be
\sum_{i=1}^n \abs{\tilde{\gamma}^1_{R_i}}\le \prod_{i=1}^n 2^{\abs{\tilde{\gamma}^1_{R_i}}},
\ee
the following expression is an upper bound for $\abs{\frac{d^2}{dt^2}\prod_{i=1}^n \zeta^t_{\beta}(R_i)}$,
\begin{align*}
& \sum_{i=1}^n \left[ \sum_{\substack{j=1 \\ j\neq i}}^n  \abs{\frac{d}{dt}\zeta^t_{\beta}(R_i)}\abs{\frac{d}{dt}\zeta^t_{\beta}(R_j)} \lf \prod_{\substack{\ell=1 \\ \ell\neq i,j}}^n \abs{\zeta^t_{\beta}(R_\ell)} \ri+ \abs{\frac{d^2}{dt^2}\zeta^t_{\beta}(R_i)}\prod_{\substack{\ell=1 \\ \ell\neq i}}^n \abs{\zeta^t_{\beta}(R_\ell)}  \right] \\
&\le \frac{\lVert f\rVert^2}{D_k}\lf\prod_{i=1 }^n (\delta\lVert f \rVert)^{\abs{\tilde{\gamma}^1_{R_i}}}\hat{\zeta}_{\beta}(\gamma^2_{R_i}) \ri \Bigg[ (\delta\lVert f \rVert)^{-2}\lf \sum_{i=1}^n \abs{\tilde{\gamma}^1_{R_i}} \ri^2 \\
& \quad  +\lf(\delta\lVert f \rVert)^{-1} - (\delta\lVert f \rVert)^{-2}\ri\sum_{i=1}^n \abs{\tilde{\gamma}^1_{R_i}} \Bigg]\\
&\le 
\frac{\lVert f\rVert^2}{D_k}(\delta\lVert f \rVert)^{-1} \prod_{i=1}^n (4\delta\lVert f\rVert)^{\abs{\tilde{\gamma}^1_{R_i}}}\hat{\zeta}_{\beta}(\gamma^2_{R_i}).
\end{align*}

We can apply Proposition \ref{theoremlemma1} with the choice of $C=4\delta \lVert f\rVert$. There exists $\delta_3<\delta_2$ such that, for every $\delta<\delta_3$ we have $Ce<1$, and there exists $\beta_C<\min\{\beta_{\delta},\beta_0\}$ such that, for every $\beta<\beta_{C}$,
\begin{align*}
&\sum_{n=1}^{\infty}\sum_{(R_{1}, \ldots, R_{n})\in \mathcal{L}^n}\abs{\phi^{T}\left(R_{1}, \ldots, R_{n}\right)}\abs{\frac{d^{2}}{d t^{2}} \prod_{i=1}^{n} \zeta^t_{\beta}\left(R_{i}\right)\Bigg\vert_{t=\theta}}\\
&\le \frac{\lVert f\rVert^2}{D_k}(\delta\lVert f \rVert)^{-1}\sum_{n=1}^{\infty}\sum_{(R_{1}, \ldots, R_{n})\in \mathcal{L}^n}\abs{\phi^{T}\left(R_{1}, \ldots, R_{n}\right)} \prod_{i=1}^{n} (4\delta \lVert f\rVert)^{\abs{\tilde{\gamma}^1_{R_i}}}\hat{\zeta}_{\beta}(\gamma^2_i)\\
&\le \frac{\lVert f\rVert^2}{D_k}(\delta\lVert f \rVert)^{-1} a_{\beta} \abs{\Lambda_k}.
\end{align*}

Applying Proposition \ref{prop1} item (a), for every $\beta<\beta_C$, we get the following inequality,
\be
\operatorname{Re} \frac{d^{2}}{d t^{2}} \zeta^t_{\beta}(\{y\})\Bigg\vert_{t=\theta}
\le -\frac{\cos(\delta \lVert f\rVert)}{D_k}d(\beta_C) \qquad \text{for every } y\in \Lambda_k.
\ee
We have
\be
\abs{ \mu^{\omega_k}_{\Lambda_k,\beta} \left(\exp \left(\frac{i t}{ \sqrt{D_{k}}} \sum_{x \in \Lambda_{k}}f\circ \theta_x(\sigma)\right) \right) }
\le \exp\lf -t^2 D \frac{\abs{\Lambda_k}}{D_k} \ri,
\ee
where we can choose $\delta<\delta_3$ sufficiently small so that the constant $D$ defined as
\be
D=\frac{1}{2}\lf \cos(\delta \lVert f\rVert)d(\beta_C) -\delta\lVert f\rVert^3 - B(\delta)\lVert f\rVert^2 -\lVert f\rVert^2(\delta\lVert f \rVert)^{-1} a_{\beta} \ri>0
\ee
is positive for a sufficiently small $\beta=\beta(\delta)<\beta_C$.
\end{proof}

\begin{lemma}\label{lem23}
Suppose that $\Phi$ is a translation invariant and absolutely summable potential, and $f:\Omega \to \mathbb{Z}$ is a bounded $\scrf_0$-measurable function satisfying $\Var_{\lambda}(f)>0$ and $f\circ \theta_x$ is a lattice distributed random variable for every $x\in \mathbb{Z}^d$. For every $0<\delta<\pi/h$ there exist $\beta(c)>0$ and  a positive constant $C=C(\delta,\beta(c))$, not depending on $\omega_k$ and $\Lambda_k$, such that, for every $\beta<\beta(c)$ and $\delta \sqrt{D_{k}} \leq \abs{t} \leq \pi \sqrt{D_{k}}/h$,
\begin{equation}\label{eq:lem23}
\abs{\mu^{\omega_k}_{\Lambda_k,\beta}\left(\exp \left(i t\bar{S}_k\right)\right)} \leq \exp \left(-C\abs{\Lambda_{k}}\right).
\end{equation}
\end{lemma}

\begin{proof}

By Fubini's Theorem, let us write the partition function $Z^{\omega_k}_{\Lambda_k,\beta,t}$ as
\be
Z^{\omega_k}_{\Lambda_k,\beta,t}
=\left(\prod_{x\in \Lambda_k} \int_E e^{-\beta h_x^{\omega_k}(\sigma_x)}\lambda(\diff \sigma_x) \right)\tilde{\Xi}^{\omega_k}_{\Lambda_k,\beta,t}
\ee
where
\be
\tilde{\Xi}^{\omega_k}_{\Lambda_k,\beta,t}=\sum_{n=0}^{\infty}\sum_{\substack{\{R_1,\ldots,R_n\} \\ \tilde{R}_i\cap \tilde{R}_j =\emptyset, i\neq j}}\prod_{x\in \Lambda_k \setminus \cup_{i=1}^n \tilde{R}_i}\left[\mathbb{E}^{\omega }_x\left( \exp\left(it\frac{f\circ \theta_x}{\sqrt{D_k}}\right)\right)\right]\prod_{i=1}^n \tilde{\zeta}^t_{\beta}(R_i),
\ee
the polymers $R_i\in \mathcal{R}_2$, and the activity function $\tilde{\zeta}^t_{\beta}$ is defined by
\be
\tilde{\zeta}^t_{\beta}(R)=\int_{E^{\tilde{R}}} \prod_{x\in\tilde{R}}\left[\exp\left(it\frac{f\circ \theta_x}{\sqrt{D_k}}\right)p_x^{\omega_k}(\sigma_x)\right]\prod_{\{x,y\}\in R}\xi_{\beta,\{x,y\}}(\sigma)\prod_{x\in \tilde{R}} \lambda(\diff \sigma_x).
\ee
Note that $\abs{\tilde{\zeta}^t_{\beta}(R)}\le \hat{\zeta}_{\beta}(R)$ for every $t\in \mathbb{R}$, where $\hat{\zeta}_{\beta}$ is defined in (\ref{hatzeta}).

Consider $c>0$ from Proposition \ref{prop1}, item (c).
Thus, for every $t$ in the interval $\delta \sqrt{D_k} \leq \abs{t} \leq \pi \sqrt{D_k}/h$ and for every $\beta<\beta'_{\delta}$,
\be
\abs{\tilde{\Xi}^{\omega_k}_{\Lambda_k,\beta,t}}\le e^{-c\abs{\Lambda_k}}\lf 1+\sum_{n=1}^{\infty}\sum_{\substack{\{R_1,\ldots,R_n\} \\ \tilde{R}_i\cap \tilde{R}_j =\emptyset, i\neq j}}\prod_{i=1}^n\eta^{c}_{\beta}(R_i)\ri =e^{-c\abs{\Lambda_k}} \Xi^{\omega_k}_{\Lambda_k,\beta}\left(\eta^{c}_{\beta}\right),
\ee
where $\eta^c_{\beta}(R)=e^{c\abs{\tilde{R}}}\hat{\zeta}_{\beta}(R)$ and $\Xi^{\omega_k}_{\Lambda_k,\beta}(\eta^{c}_{\beta})$ is defined in (\ref{xiforeta}).

By Theorem \ref{theoremeta}, for every $\beta<\beta_c$, the function $\Xi^{\omega_k}_{\Lambda_k,\beta}(\eta^{c}_{\beta})$ can be written as in  (\ref{xietathm}).
Note that $\tilde{\Xi}^{\omega_k}_{\Lambda_k,\beta,0}=\Xi^{\omega_k}_{\Lambda_k,\beta}$, i.e., we take $t=0$ in (\ref{cluster}). By Theorem \ref{clusterexpansionthm}, there exists $\alpha_{\beta}<1$ that decreases to 0 when $\beta \to 0$ such that
 \be\label{alpha0}
\sum_{n=1}^{\infty} \sum_{(R_1,\ldots,R_n)\in \mathcal{R}^n_2} \abs{\phi^T(R_1,\ldots,R_n)} \prod_{i=1}^n \abs{\zeta_{\beta}(R_i)} \le \alpha_{\beta}\abs{\Lambda_k}.
 \ee
Therefore, by (\ref{removebar}), (\ref{etathm}), and (\ref{alpha0}),
\begin{align*}
&\abs{\mu^{\omega_k}_{\Lambda_k,\beta}\left(\exp \left(i t \bar{S}_k\right)\right)} \leq  e^{-c\abs{{\Lambda}_{k}}} \abs{\frac{\Xi^{\omega_k}_{\Lambda_k,\beta}\left(\eta^{c}_{\beta}\right)} {\Xi^{\omega_k}_{\Lambda_k,\beta}}} \\
&\le \exp\left(-c\abs{\Lambda_k}+\sum_{n=1}^{\infty}\sum_{(R_1,\ldots,R_n)\in \mathcal{R}^n_2} \abs{\phi^T(R_1,\ldots, R_n)} \abs{\prod_{i=1}^n\eta^c_{\beta}(R_i)-\prod_{i=1}^n\zeta_{\beta}(R_i)} \right)\\
&\le \exp\left(-c\abs{\Lambda_k}+\sum_{n=1}^{\infty}\sum_{(R_1,\ldots,R_n)\in \mathcal{R}^n_2} \abs{\phi^T(R_1,\ldots, R_n)} \left(\prod_{i=1}^n \eta^c_{\beta}(R_i)+\prod_{i=1}^n\abs{\zeta_{\beta}(R_i)} \right)\right)\\
&\le e^{(-c+\alpha_{\beta}+\bar{\alpha}_{c,\beta})\abs{\Lambda_k}}.
\end{align*}
Then, Estimate (\ref{eq:lem23}) holds when $\beta$ is sufficiently small.
\end{proof}

\subsection{Proof of Theorem \ref{main} Condition (\ref{samecampanino})}

The proof of Theorem \ref{main} condition (\ref{samecampanino}) is similar to the proof of the main result in \cite{CCT} since we apply their conditions for the potential to have absolute convergence of the series of the cluster expansion. The third and the fourth integrals in (\ref{diff}) will be small due to Lemmas \ref{lemcam1} and \ref{lemcam2} below, in which the proofs are similar to Lemmas \ref{lem22} and \ref{lem23}, respectively. Here, we are going to explain some computation already done in \cite{CCT} to be clear when we adapt to a family of potentials satisfying (\ref{campanino}).

Theorem \ref{cct} below is Theorem $3.2$ in \cite{CCT} that refers to the convergence of the series of the cluster expansion that is valid for a family of models. For this, let us list some definitions.

For the set of polymers $\mathcal{R}$ defined in Section \ref{ce} and an activity function $\kappa:\mathcal{R}\to \mathbb{C}$ such that $\abs{\kappa(R)}$ is bounded, define the partition function for a gas of polymers, with activity $\kappa$, and hardcore interaction in $\Lambda\Subset \mathbb{Z}^d$, by
\be
\Xi_{\Lambda}(\kappa)=1+\sum_{n=1}^{\infty}\sum_{\substack{\{R_1,\ldots,R_n\} \\ \tilde{R}_i\cap \tilde{R}_j=\emptyset, i\neq j}}\prod_{i=1}^n \kappa(R_i).
\ee

\begin{theorem}[Campanino, Capocaccia, Tirozzi -- CMP 1979]\label{cct}
Let $\Psi$ be a real, positive function on $\mathbb{Z}^d$ such that $\Psi(0)=1$,
\begin{equation}
\sum_{x\in \mathbb{Z}^d} \Psi(x)^{1/2}=K<\infty,
\end{equation}
and $z_0$ a positive number such that $\sqrt{z_0}K<1$. For $b\in\mathcal{P}_{1,2}$, define \be
\tilde{\Psi}(b)=
\begin{cases}
1 &\text{ if } \abs{b}=1 \\
\sup_{x,y\in b}\Psi(x-y) &\text{ if } b=\{x,y\}
\end{cases}.
\ee
Assume that
\begin{equation}
\abs{\kappa(R)}\le \prod_{b\in R}z_0 \tilde{\Psi}(b).
\end{equation}
Then, for every $x\in \mathbb{Z}^d$,
\be
\sum_{R: \tilde{R}\ni x}\abs{\kappa(R)}\le
\sum_{R: \tilde{R}\ni x} \prod_{b\in R}z_0 \tilde{\Psi}(b)\le
\frac{\sqrt{z_0}K}{1-\sqrt{z_0}K}:=B(z_0,K).
\ee
Moreover, if
\be
z_0\exp\lf B(\sqrt{z_0},K)\ri=C(z_0,K)<1,
\ee
then, for every polymer $R\in \mathcal{R}$,
\be\label{camcondition}
 \sum_{n=1}^{\infty}\sum_{\substack{(R_1,\ldots, R_n)\\ \exists R_i=R}}\abs{\phi^T(R_1,\ldots,R_n)}\prod_{i=1}^n \abs{\kappa(R_i)}\le \frac{C(z_0,K)}{1-C(z_0,K)}\frac{\abs{\kappa(R)}}{ \prod_{b\in R} \sqrt{z_0}}
\ee
and
\be
 \sum_{n=1}^{\infty}\sum_{\substack{(R_1,\ldots, R_n)\\ \exists R_i=R}}\abs{\phi^T(R_1,\ldots,R_n)}\prod_{i=1}^n \prod_{b\in R_i}z_0 \tilde{\Psi}(b)\le \frac{C(z_0,K)}{1-C(z_0,K)}\prod_{b\in R}  \sqrt{z_0}\tilde{\Psi}(b).
\ee
Under Condition (\ref{camcondition}), we have
\be
\Xi_{\Lambda}(\kappa) = \exp \lf \sum_{n=1}^{\infty}\sum_{(R_1,\ldots,R_n)} \phi^T(R_1,\ldots,R_n)\prod_{i=1}^n \kappa(R_i) \ri.
\ee
\end{theorem}


For fixed positive integer $r_0$, define
\be
\mathbb{Z}^{d}(r_0)=\{(n_1r_0,\ldots,n_dr_0): n_1,\ldots, n_d\in \mathbb{Z}\}
\ee
to be the sublattice of $\mathbb{Z}^d$, and $\Lambda^{r_0}_k = \Lambda_k\cap \mathbb{Z}^d(r_0)$. The main idea in \cite{CCT} is to control the characteristic function $\mu^{\omega}_{\Lambda_k,\beta}(e^{it\bar{S}_k})$ by taking $r_0$ large enough, since the spins in $\Lambda^{r_0}_k$ will be ``almost independent". Thus, this approach does not require to consider large temperatures. For a configuration $\omega'\in \Omega$, let us bound the characteristic function as below,
\begin{align*}
&\abs{\mu^{\omega_k}_{\Lambda_k,\beta}\lf e^{it\bar{S}_k}\ri} 
= \abs{\mu^{\omega_k}_{\Lambda_k,\beta}\lf\mu^{\omega_k}_{\Lambda_k,\beta}\lf \exp\lf\frac{it}{\sqrt{D_k}}S_k\ri \Bigg\vert \sigma_x=\omega'_x, x\in\Lambda_k\setminus \Lambda^{r_0}_k\ri \ri }\\
&\le \sup_{\omega' \in \Omega}\abs{\mu^{\omega_k}_{\Lambda_k,\beta}\lf \exp\lf\frac{it}{\sqrt{D_k}}\sum_{x\in \Lambda^{r_0}_k}f\circ \theta_x(\sigma)\ri \Bigg\vert \sigma_x=\omega'_x, x\in\Lambda_k\setminus \Lambda^{r_0}_k\ri }.
\end{align*}

By spatial Markov property (see \cite{FV}, Section 3.6.3), we have
\begin{align*}
&\mu^{\omega_k}_{\Lambda_k,\beta}\lf \exp\lf \frac{it}{\sqrt{D_k}}\sum_{x\in \Lambda^{r_0}_k}f\circ \theta_x(\sigma)\ri \Bigg\vert \sigma_x=\omega'_x, x\in\Lambda_k\setminus \Lambda^{r_0}_k \ri\\
&= \mu^{(\omega_k\lor \omega')^{r_0}}_{\Lambda^{r_0}_k,\beta}\lf \exp\lf \frac{it}{\sqrt{D_k}}\sum_{x\in \Lambda^{r_0}_k}f\circ \theta_x(\sigma)\ri \ri,
\end{align*}
where the boundary condition $(\omega_k\lor \omega')^{r_0}$ is defined by
\be
(\omega_k\lor \omega')^{r_0}(x)=
\begin{cases}
\omega'(x) &\text{ if }x\in \Lambda_k\setminus \Lambda^{r_0}_k\\
\omega_k(x) &\text{ if }x\in \Lambda^c_k
\end{cases}.
\ee

\begin{lemma}\label{lemcam1}
Suppose that $\Phi$ is a translation invariant and absolutely summable potential satisfying Condition (\ref{campanino}), and $f:\Omega \to \mathbb{Z}$ is a bounded $\scrf_0$-measurable function satisfying $\Var_{\lambda}(f)>0$.
For every $\beta>0$, there exists $\delta_{\beta}>0$ such that, for every $\delta<\delta_{\beta}$, there exists $r_0(\delta,\beta)>0$ satisfying the following: For every $r_0\ge r_0(\delta,\beta)$, there exists a positive constant $D=D(\beta,\delta,r_0)$ such that, if $\abs{t}< \delta \sqrt{D_k}$, then
\be
 \abs{ \mu^{(\omega_k\lor \omega')^{r_0}}_{\Lambda^{r_0}_k,\beta}\lf \exp\lf \frac{it}{\sqrt{D_k}}\sum_{x\in \Lambda^{r_0}_k}f\circ \theta_x(\sigma)\ri \ri }
 \le \exp\lf -t^2D\frac{\abs{\Lambda^{r_0}_k}}{D_k} \ri
 \ee
 uniformly with respect to $\omega'\in \Omega$.
\end{lemma}

\begin{proof}
For a fixed $r_0>0$ and $\beta>0$, define $\bar{\Phi}(r_0):=\sup_{\lVert x\rVert\ge r_0}\lVert \Phi_{\{x,0\}}\rVert>0$. For every $x,y\in \mathbb{Z}^d(r_0)$, we have
\be
\abs{\xi_{\beta,\{x,y\}}(\sigma)}\le \beta \lVert \Phi_{\{x,y\}} \rVert e^{\beta \bar{\Phi}(r_0)}.
\ee
Define 
\be\label{psigood}
\Psi(x):=
\begin{cases}
\frac{\lVert \Phi_{\{x,0\}} \rVert}{\bar{\Phi}(r_0)} &\text{ if }x\in \mathbb{Z}^d(r_0)\setminus\{0\}\\
1 &\text{ if }x=0
\end{cases}.
\ee
Note that
\be
K=\sum_{x\in \mathbb{Z}^d(r_0)}\Psi(x)^{1/2}\le 1+ \frac{1}{\bar{\Phi}(r_0)^{1/2}}\sum_{x\neq 0}\lVert \Phi_{\{x,0\}}\rVert^{1/2}<\infty.
\ee
For fixed polymer $R\in \mathcal{R}$ with $\abs{\tilde{R}}\ge 2$, we have
\be\label{ccthatzeta}
\hat{\zeta}_{\beta}(\gamma^2_R) \le \prod_{\{x,y\}\in R}  \beta \lVert \Phi_{\{x,y\}} \rVert e^{\beta \bar{\Phi}(r_0)}.
\ee
Define
\be
z_0:=\max\left\{4\delta \lVert f\rVert, \beta \bar{\Phi}(r_0)e^{\beta \bar{\Phi}(r_0)} \right\}.
\ee
For every sufficiently small $\delta$ and $\abs{t}<\delta \sqrt{D_k}$,
\be
\abs{\frac{d^2}{dt^2} \prod_{i=1}^n \zeta^t_{\beta}(R_i)}
\le \frac{\lVert f\rVert^2}{D_k}(\delta \lVert f\rVert)^{-1}\prod_{i=1}^n \prod_{b\in R_i}z_0\tilde{\Psi}(b). 
\ee
By Theorem \ref{cct}, for $r_0$ sufficiently large,
\begin{align*}
&\sum_{n=1}^{\infty}\sum_{(R_1,\ldots,R_n)\in \mathcal{L}^n} \abs{\phi^T(R_1,\ldots,R_n)}\abs{\frac{d^2}{dt^2} \prod_{i=1}^n \zeta^t_{\beta}(R_i)}\\
&\le \frac{\lVert f\rVert^2}{D_k}(\delta \lVert f\rVert)^{-1} \sum_{\substack{R\in \mathcal{R} \\ \abs{\tilde{R}}\ge 2}} \sum_{n=1}^{\infty}\sum_{\substack{(R_1,\ldots,R_n) \\ \exists R_i=R}}\abs{\phi^T(R_1,\ldots,R_n)}\prod_{i=1}^n \prod_{b\in R_i}z_0\tilde{\Psi}(b)\\
&\le \frac{\lVert f\rVert}{D_k}\frac{z_0}{\delta} \frac{\exp(B(\sqrt{z_0},K))}{1-C(z_0,K)}B(\sqrt{z_0},K)\abs{\Lambda^{r_0}_k}.
\end{align*}

Note that the constant
\be
\varphi_{\beta}(\delta,r_0)= \lVert f \rVert\frac{z_0}{\delta} \frac{\exp(B(\sqrt{z_0},K))}{1-C(z_0,K)}B(\sqrt{z_0},K)
\ee
goes to zero when $z_0$ is sufficiently small. More precisely, for a fixed $\delta$,  such that, for every $r_0\ge r'_0(\delta,\beta)$, we have $z_0=4\delta \lVert f\rVert$. Then, $\varphi_{\beta}(\delta,r_0)\to 0$, when $\delta$ and $r^{-1}_0$ are small enough.  The rest of the proof follows the same as in the proof of Lemma \ref{lem22}. 

Thus, there exists $\delta_{\beta}$ such that, for $\delta>\delta_{
\beta}$ and $r_0(\delta,\beta)> r'_0(\delta,\beta)$ so that the constant $D$ defined as
\be
D=\frac{1}{2}\lf \cos(\delta \lVert f\rVert)d(\beta) -2\delta\lVert f\rVert^3 - B(\delta)\lVert f\rVert^2 - \varphi_{\beta}(\delta,r_0)\ri
\ee
is positive for every $r_0\ge r_0(\delta,\beta)$.
\end{proof}

\begin{lemma}\label{lemcam2}
Suppose that $\Phi$ is a translation invariant and absolutely summable potential satisfying the condition (\ref{campanino}), and $f:\Omega \to \mathbb{Z}$ is a bounded $\scrf_0$-measurable function satisfying $\Var_{\lambda}(f)>0$ and $f\circ \theta_x$ is a lattice distributed random variable for every $x\in \mathbb{Z}^d$.
For every $\beta>0$ and $\delta>0$, there exists $r_1(\delta,\beta)$ such that, for every $r_0\ge r_1(\delta,\beta)$, there exists a positive constant $C$, not depending on the sequence of the boundary condition $\omega_k$, such that, if $\delta \sqrt{D_k}\le \abs{t}\le \pi\sqrt{D_k}/h$, then
\be
\abs{ \mu^{(\omega_k\lor \omega')^{r_0}}_{\Lambda^{r_0}_k,\beta}\lf \exp\lf \frac{it}{\sqrt{D_k}}\sum_{x\in \Lambda^{r_0}_k}f\circ \theta_x(\sigma)\ri \ri } \le \exp\lf -C\abs{\Lambda^{r_0}_k} \ri
\ee 
uniformly with respect to $\omega'\in \Omega$.
\end{lemma}

\begin{proof}
Consider $\Psi(x)$ as in (\ref{psigood}). For $\beta>0$, choosing 
\be
z_0:=\beta \bar{\Phi}(r_0)e^{2c+\beta \bar{\Phi}(r_0)} \quad \text{and} \quad
z_1:=\beta \bar{\Phi}(r_0)e^{\beta \bar{\Phi}(r_0)},
\ee
where $c>0$ is the constant in Proposition \ref{prop1} item (b), and taking $r_0$ sufficiently large, by Theorem \ref{cct} and (\ref{ccthatzeta}),
\begin{align*}
 \sum_{n=1}^{\infty}\sum_{(R_1,\ldots,R_n)\in \mathcal{R}^n_2}\abs{\phi^T(R_1,\ldots,R_n)}\prod_{i=1}^n \eta^c_{\beta}(R_i) &\le \frac{C(z_0,K)}{1-C(z_0,K)}B(\sqrt{z_0},K) \abs{\Lambda^{r_0}_k},\\
  \sum_{n=1}^{\infty}\sum_{(R_1,\ldots,R_n)\in \mathcal{R}^n_2}\abs{\phi^T(R_1,\ldots,R_n)}\prod_{i=1}^n \abs{\zeta_{\beta}(R_i)} &\le \frac{C(z_1,K)}{1-C(z_1,K)}B(\sqrt{z_1},K) \abs{\Lambda^{r_0}_k},
\end{align*}
where $\eta^c_{\beta}(R)=e^{c\abs{\tilde{R}}}\hat{\zeta}_{\beta}(R)$.
The proof follows the same strategy as the proof of Lemma \ref{lem23}, where there exists $r_1(\delta,\beta)$ such that the constant
\be
C=c-\frac{C(z_0,K)}{1-C(z_0,K)}B(\sqrt{z_0},K) -\frac{C(z_1,K)}{1-C(z_1,K)}B(\sqrt{z_1},K)
\ee
is positive for every $r_0\ge r_1(\delta,\beta)$.
\end{proof}

From Lemmas \ref{lemcam1} and \ref{lemcam2}, and using the bound (\ref{diff}), we conclude the proof of Theorem \ref{main}, condition (\ref{samecampanino}).

\section*{Acknowledgements}
The authors would like to kindly thank  Roberto Fern\'andez and Tong Xuan Nguyen for very useful discussions. We thank the referees and Aernout van Enter for their suggestions that helped us clarify our manuscript. Also, both authors would like to acknowledge the support of the NYU-ECNU Institute of Mathematical Sciences at NYU Shanghai. 
%
%
%

\end{document}